\documentclass[12pt,a4paper]{article}
\usepackage{amsfonts,amsmath,amssymb,enumitem,fullpage,graphicx,mathrsfs}

\newtheorem{corollary}{Corollary}[section]
\newtheorem{definition}[corollary]{Definition}
\newtheorem{example}[corollary]{Example}
\newtheorem{lemma}[corollary]{Lemma}

\newtheorem{proposition}[corollary]{Proposition}

\newtheorem{theorem}[corollary]{Theorem}

\newcommand{\qed}{\rule{.07in}{.1in}}
\newenvironment{proof}{\vspace{1ex}\noindent\textbf{Proof}\hspace{0.5em}}{\hfill\qed\vspace{1ex}}

\title{Pseudocodeword-Free Criterion for Codes with Cycle-Free Tanner Graph}
\author{Wittawat Kositwattanarerk
\thanks{W. Kositwattanarerk is with the Department of Mathematics, Faculty of Science, Mahidol University, Bangkok 10400, Thailand and the Centre of Excellence in Mathematics, the Commission on Higher Education, Bangkok 10400, Thailand (e-mail: wittawat.kos@mahidol.edu).}}
\date{\today}

\begin{document}

\maketitle

\begin{abstract}
Iterative decoding and linear programming decoding are guaranteed to converge to the maximum-likelihood codeword when the underlying Tanner graph is cycle-free. Therefore, cycles are usually seen as the culprit of low-density parity-check (LDPC) codes. In this paper, we argue in the context of graph cover pseudocodeword that, for a code that permits a cycle-free Tanner graph, cycles have no effect on error performance as long as they are a part of redundant rows. Specifically, we characterize all parity-check matrices that are pseudocodeword-free for such class of codes.
\end{abstract}



\section{Introduction}

Modern decoding algorithms such as message-passing iterative decoding and linear programming decoding are extremely efficient and are shown to enable communications at rates near the channel capacity under several circumstances. These decoders are known to converge when the Tanner graph is cycle-free, and so one of the design criteria for the constructions of LDPC codes is the number and size of cycles in the underlying graphs. Thus, regular and irregular LDPC codes are usually constructed semi-randomly with procedures that avoid cycles with small girth \cite{RSU,Tian,Xu}. There are also investigations on the effects of cycles on the performance of iterative decoders \cite{L}.


Another explanation for when iterative decoders fail to converge is by means of the pseudocodewords. Since the pseudocodewords satisfy every condition set by the decoder, they are legitimate from the perspective of the algorithm. The so-called pseudoweight acts as Hamming weight for LDPC codes, and the pseudocodewords provide a tangible framework for the study of the error performance of LDPC codes. For this reason, these noncodeword outputs and their properties have been extensively studied in the literature \cite{Axvig,KS,KLVW,KM,XF,ZSF}.

Wiberg \cite{W} is among the first to study noncodeword outputs from iterative decoders where computations from the algorithms are retracted and laid out as a tree. Since iterative decoders perform calculations locally, it is probable that the algorithms try to make an estimate on a graph that behaves locally like the original Tanner graph. In \cite{KLVW}, the pseudocodewords are characterized using a finite degree lift of the Tanner graph called a graph cover, and this description of the pseudocodeword relates well to noncodeword outputs from linear programming decoding \cite{FWK}. Excellent overviews of the pseudocodewords can be found, for example, in \cite{Axvig,KS}.

In this paper, we focus on the pseudocodewords arising from graph covers of the Tanner graph. We provide an exact condition for codes with cycle-free Tanner graph to be pseudocodeword-free. Although codes in this class are known to have limited capabilities \cite{ETV}, the results given here yield a surprising insight on the structure of the pseudocodewords: good representation for this class of codes has \textit{nothing} to do with cycles as long as there exists a spanning tree of the Tanner graph that represents the same code. As a result, this work sheds light on empirical phenomenons where, under certain circumstances, iterative decoders perform well despite a number of small cycles in the Tanner graph \cite{KLF} and eliminating small cycles does not significantly improve decoding performance of LDPC codes \cite{MN}.

The remainder of this paper is organized as follows. Section 2 provides some background on parity-check codes and the pseudocodewords. We introduce the notion of p-satisfy in Section 3 to  ease our searches for the pseudocodewords. The main result of this work is stated as Theorem \ref{th:main}. Several examples are given in Section 4.

\section{Preliminaries}

A code can be represented by many parity-check matrices, but a parity-check matrix uniquely determines a code. This choice of representation can be especially problematic in the context of decoding algorithms that operate on a parity-check matrix. To avoid unnecessary confusions, we define a code based on its parity-check matrix.

\begin{definition}
Let $H\in\mathbb{F}_2^{r\times n}$. The binary linear code with parity-check matrix $H$, denoted $C(H)$, is the null space of $H$. In other words,
\[C(H)=\{\mathbf{y}\in\mathbb{F}_2^n\mid H\mathbf{y}^T=\mathbf{0}\in\mathbb{F}_2^{r\times 1}\}.\]
\end{definition}

Here, we do not impose that rows of $H$ are linearly independent. The above definition coincides with the usual terminology--if $H\in\mathbb{F}_2^{r\times n}$, then $C(H)$ is a subspace of $\mathbb{F}_2^n$ of dimension at least $n-r$. We shall write $C$ as a code to mean a subspace of $\mathbb{F}_2^n$ without a specific choice of parity-check matrix and say that $H$ represents $C$ if $C=C(H)$.
Throughout, our discussions will focus on the parity-check matrix $H$ and not just the code $C$, since iterative decoders and linear programming decoder depend rather on the choice of parity-check matrix of the code. The Tanner graph from a binary matrix is defined next.

\begin{definition}
Let $H\in\mathbb{F}_2^{r\times n}$. The \textit{Tanner graph} of $H$, denoted $T(H)$, is a bipartite graph with biadjacency matrix $H$. Specifically,
\[T(H)=(X\cup F,E)\]
where $X=\{x_1,\dots,x_n\}$ represents the columns of $H$ and is called the set of \textit{bit nodes}, $F=\{f_1,\dots,f_r\}$ represents the rows of $H$ and is called the set of \textit{check nodes}, and
\[E=\{\{x_i,f_j\}\mid h_{ji}=1\}.\]
\end{definition}

Given a parity-check matrix $H$, it is clear that the Tanner graph is bipartite. The converse of this statement is also true: given a bipartite graph, it can be viewed as a Tanner graph of some parity-check matrix. In a way, the Tanner graph gives a graphical representation of the parity-check matrix. Suppose that binary values $c_1,c_2,\ldots,c_n$ are assigned to the bit nodes of the Tanner graph $T(H)$. Then, $\mathbf{c}=(c_1,c_2,\ldots,c_n)$ is a codeword of $C(H)$ if and only if the binary sum of the values at the neighbors of every check node is zero. Next, we give the definition of a graph cover.

\begin{definition}
Let $m$ be a positive integer. An $m$\textit{-cover} of $T(H)$ is a graph $\widehat{T}(H)$ with the property that there exists an $m$-to-1 mapping $\varphi$ from the vertices of $\widehat{T}(H)$ to the vertices of $T(H)$ that preserves degree and the set of neighbors. Namely, if $v$ is a vertex of $\widehat{T}(H)$ of degree $t$ and with neighbors $u_1,\ldots,u_t$, then $\varphi(v)$ is a vertex of $T(H)$ of degree $t$ and with neighbors $\varphi(u_1),\ldots\varphi(u_t)$.
\end{definition}

Although we only study a graph cover of the Tanner graph $T(H)$, we remark here that the above definition can be applied to give a graph cover of any graph. Since a graph cover $\widehat{T}(H)$ is bipartite, it can be viewed as a Tanner graph of some parity-check matrix. Similar to how $T(H)$ can be used to determine codewords, the graph cover $\widehat{T}(H)$ can be used to determine what is called graph cover pseudocodeword.

Before we proceed to the definition of graph cover pseudocodeword, we briefly discuss iterative decodings and linear programming decoding and refer readers to \cite{FWK,KFL} for a more precise description. Suppose that information are sent over a memoryless binary-input symmetric-output channel. Given a received word $\mathbf{w}$, maximum likelihood (ML) decoder finds $\mathbf{y}\in C$ that maximizes $P(\mathbf{w}\mid\mathbf{y})$. This is equivalent to choosing $\mathbf{y}\in C$ that minimizes
\begin{equation}\label{eq:ml}
\gamma_1y_1+\gamma_2y_2+\ldots+\gamma_ny_n
\end{equation}
where
\[\gamma_i=\log\left(\frac{P(w_i\mid y_i=0)}{P(w_i\mid y_i=1)}\right)\]
is the log-likelihood ratio at the $i^{\mathrm{th}}$ coordinate. Iterative decodings (also known as message-passing iterative decoding or belief propagation) perform inference on the Tanner graph. Roughly speaking, once a word $\mathbf{w}$ is received, the log-likelihood ratios $\gamma_i$ are assigned to the bit nodes. Each bit node then broadcasts its value to the neighboring check nodes. Once a check node receives the likelihood from all its neighboring bit nodes, it uses this information to make new estimates and send them back. Each bit node updates its likelihood, and then the process iterates.

Linear programming decoding aims to minimize \eqref{eq:ml} over the convex hull of $C(H)$ when implicitly embedded in $\mathbb{R}^n$. Since the number of constraints needed to describe this convex hull is exponential in block length, linear programming decoding is typically done over a relaxed polytope
\[\bigcap_{j=1}^rconv(C(Row_j(H))),\]
which is the intersection of the convex hull of codewords from the $r$ simple parity-check codes given by the rows of $H$.

The hallmark of iterative decoding is that the entire process is local, meaning that each vertex has its own agenda and makes decision independently. Rather than searching for a codeword that satisfies every parity condition collectively, the algorithms look for a codeword that satisfies each parity condition iteratively. Thus, the algorithms cannot distinguish between the Tanner graph and its cover. This results in the algorithms trying to converge to a legitimate binary value assignment on the graph cover. We call these graph cover pseudocodewords. In a similar fashion, linear programming decoding breaks parity-check matrix into a collection of parity conditions determined by each row of the matrix. While the relaxed polytope has a more tractable representation than the original codeword polytope, it could be the case that the relaxed polytope has a vertex that is not in the convex hull of $C(H)$. Those noncodeword vertex are a scale of a graph cover pseudocodeword \cite{FWK}. From now on, we use the term pseudocodeword to refer to graph cover pseudocodeword and give its definition next.

\begin{definition}
Let $\widehat{T}(H)$ be an $m$-cover of $T(H)$ with an $m$-to-1 mapping $\varphi$. Suppose that binary values $c_{i1},c_{i2},\ldots,c_{im}$ are assigned to the preimage of the bit node $x_i$ of $T(H)$ under $\varphi$ in such a way that the binary sum of the values at the neighbors of every check node of $\widehat{T}(H)$ is zero. (In other words, $c_{11},c_{12},\ldots,c_{1m},\ldots,c_{n1},c_{n2},\ldots,c_{nm}$ is a legitimate codeword of $\widehat{T}(H)$.) Then, the integer vector
\[\left(\sum_{k=1}^m{c_{1k}},\ldots,\sum_{k=1}^m{c_{nk}}\right)\]
is a \textit{pseudocodeword} of $H$. The set of all pseudocodewords of $H$ is denoted $PC(H)$.
\end{definition}

Note that $PC(H)$ collects pseudocodewords from an $m$-cover $\widehat{T}(H)$ for all values of $m$. In particular, a $1$-cover of $T(H)$ is $T(H)$ itself, so each codeword of $C(H)$ is considered a pseudocodeword.

\begin{example}\label{ex:basic}
Consider a parity-check matrix

\[H=\left(
\begin{array}{cccc}
1 & 1 & 1 & 0 \\
0 & 1 & 0 & 1
\end{array}\right).\]
It follows that
\[C(H)=\{(0,0,0,0),(0,1,1,1),(1,0,1,0),(1,1,0,1)\}.\]
The Tanner graph $T(H)$ and a codeword $(1,1,0,1)$ on $T(H)$ is shown in Figure \ref{fig:tanner}. Figure \ref{fig:cover} portrays a 2-cover of $T(H)$ along with a pseudocodeword $(1,2,1,0)$. Here, note that $(1,2,1,0)$ is not an integral combination of elements from $C(H)$.
\end{example}

\begin{figure}
\begin{center}
\includegraphics[width=5cm]{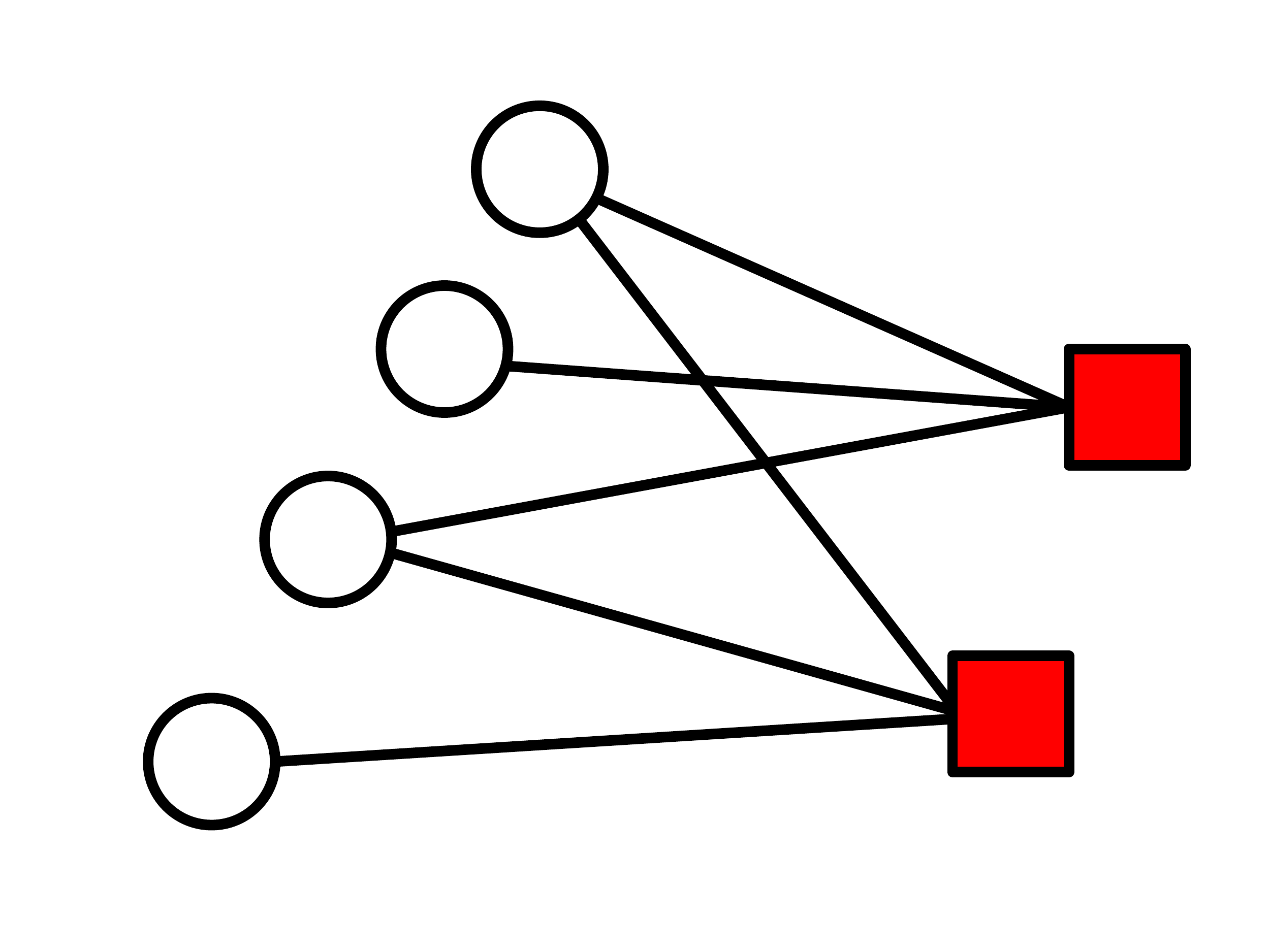} \quad\quad
\includegraphics[width=5cm]{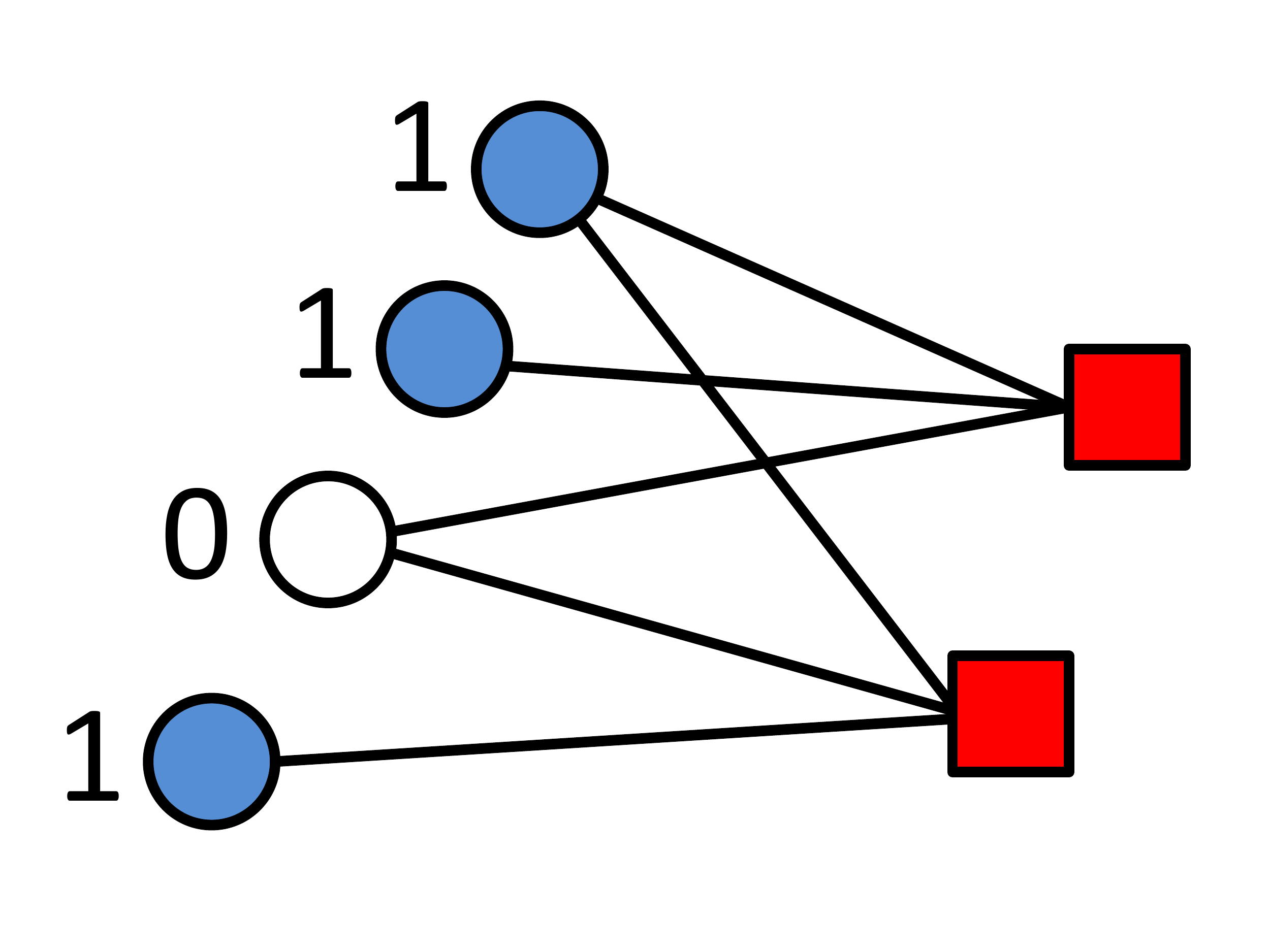}
\caption{\label{fig:tanner} On the left is the Tanner graph $T(H)$ from Example \ref{ex:basic}. White circles are bit nodes and represent columns of $H$. Red squares are check nodes and represent rows of $H$. An assignment $1,1,0,1$ to the bit nodes corresponds to the codeword $(1,1,0,1)$ of $C(H)$ and is shown on the right.}
\end{center}
\end{figure}

\begin{figure}
\begin{center}
\includegraphics[width=7cm]{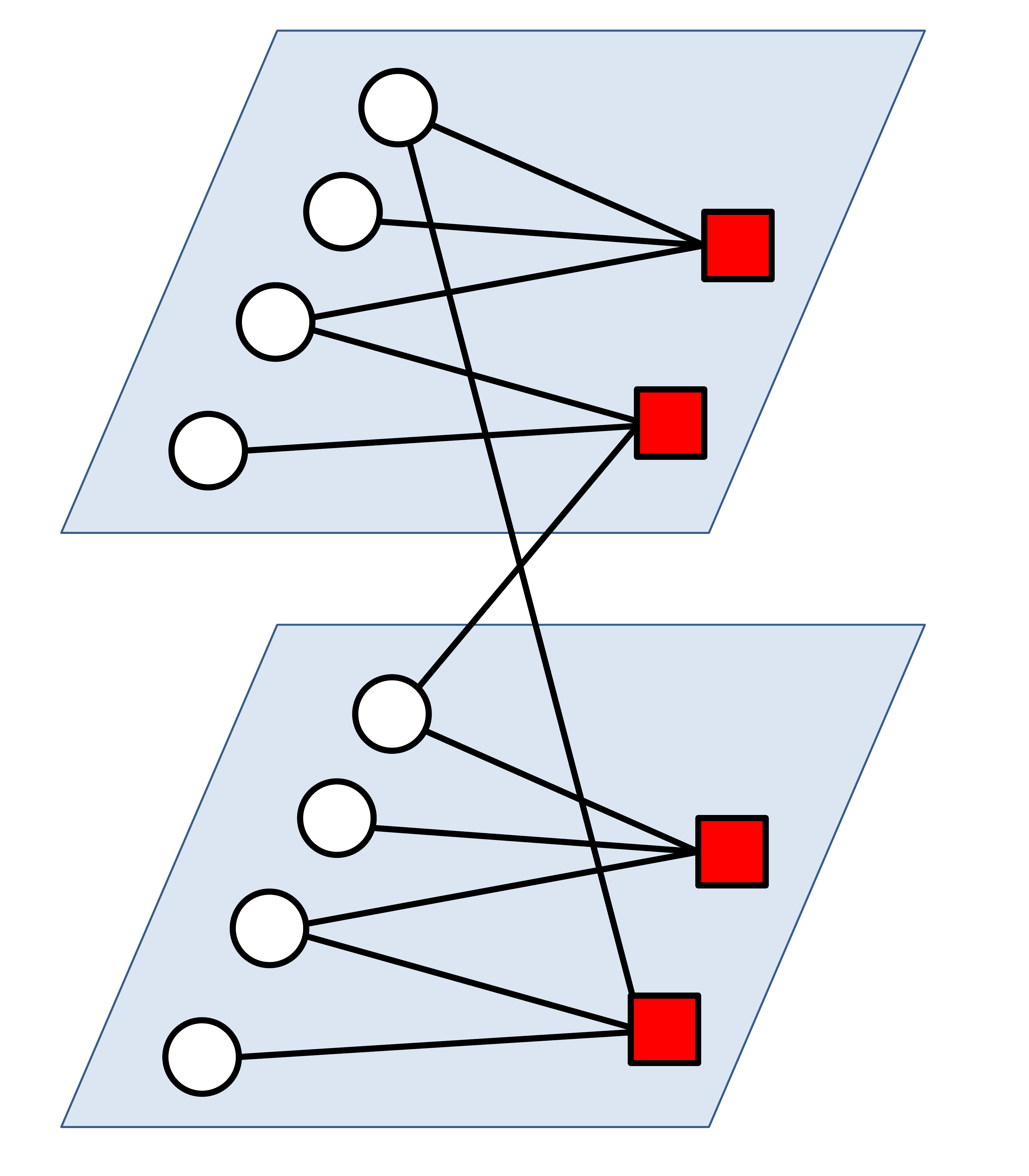} \quad\quad
\includegraphics[width=7cm]{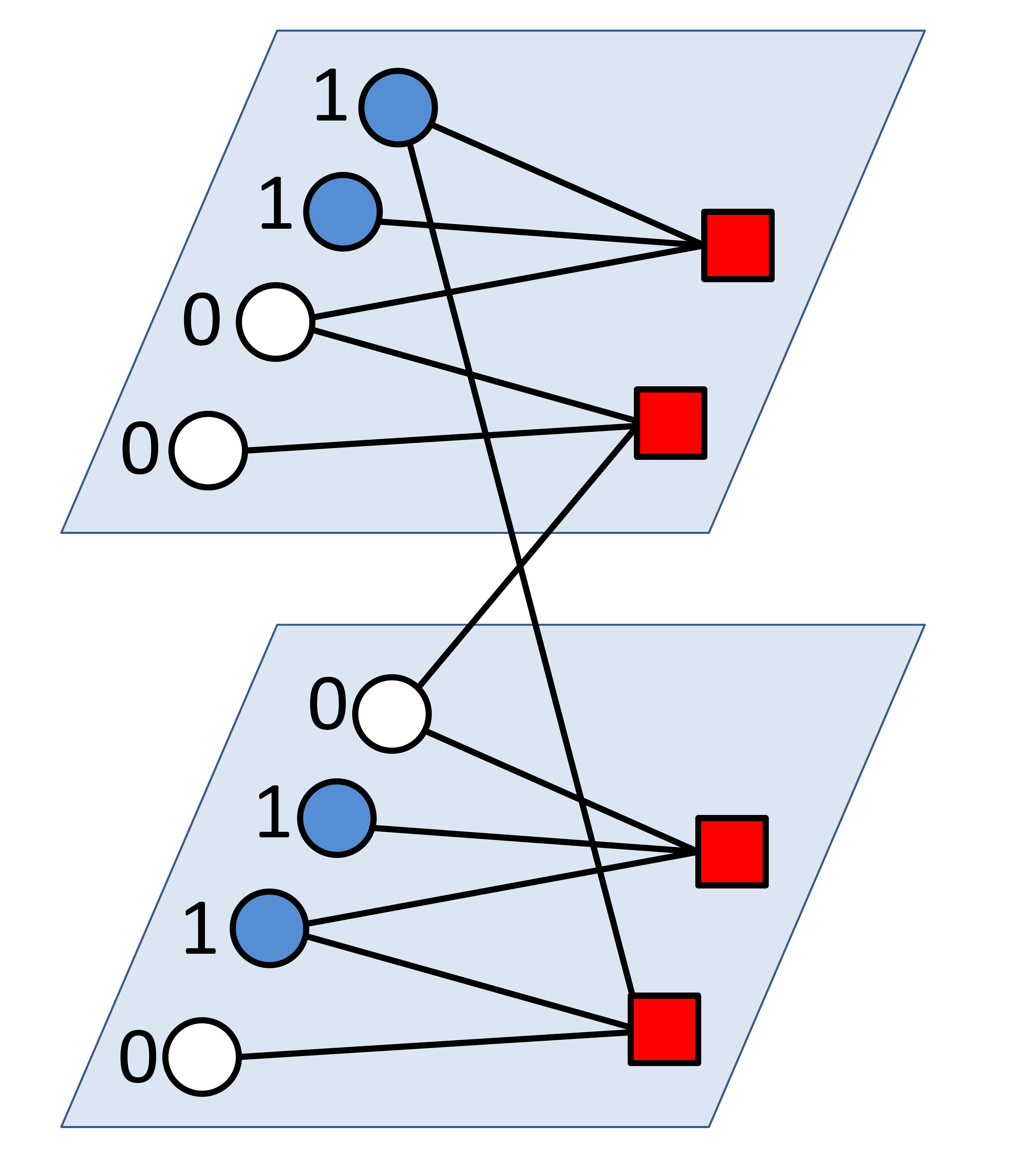}
\caption{\label{fig:cover} On the left is a 2-cover $\widehat{T}(H)$ of $T(H)$ from Example \ref{ex:basic}. A legitimate binary value assignment on the right corresponds to the pseudocodeword $(1,2,1,0)$.}
\end{center}
\end{figure}

Given a code $C$, it is desirable to represent $C$ with a parity-matrix $H$ so that the set of pseudocodewords $PC(H)$ is as small as possible. Toward this goal, we implicitly embed $C(H)$ in $\mathbb{R}^n$ and enumerate some elements of $PC(H)$. We have already seen that
\[C(H)\subseteq PC(H).\]
In fact, since a graph consisting of $m$ copies of the Tanner is consider an $m$-cover, any integral combinations of the elements of $C(H)$ whose coefficients are nonnegative are considered a pseudocodeword; that is,
\[\left\{\sum_{\mathbf{c}\in C(H)}{a_\mathbf{c}\mathbf{c}\:\bigg{|}\:a_\mathbf{c}\in\mathbb{N}}\right\}\subseteq PC(H).\]
Again, addition and multiplication here are done over $\mathbb{R}$. The case when the above inclusion becomes an equality deserves a special consideration.

\begin{definition}
A parity-check matrix $H$ is \textit{geometrically perfect} if the pseudocodewords of $H$ are precisely integer combinations of the codewords of $H$ with nonnegative coefficients. In other words,
\[PC(H)=\left\{\sum_{\mathbf{c}\in C(H)}{a_\mathbf{c}\mathbf{c}\:\bigg{|}\:a_\mathbf{c}\in\mathbb{N}}\right\}.\]
\end{definition}

Being geometrically perfect means that the set of pseudocodewords is kept as small as possible. We see that the matrix $H$ from Example \ref{ex:basic} is not geometrically perfect since $PC(H)$ contains $(1,2,1,0)$. A well-known class of geometrically perfect parity-check matrices is the collection of matrices whose Tanner graph is cycle-free. Here, if $\widehat{T}(H)$ is a cover of $T(H)$ that is cycle-free, then $\widehat{T}(H)$ simply consists of disconnected copies of $T(H)$, and so any pseudocodeword of $H$ is an integral combination of the codewords.

In this work, we tie the property of being geometrically perfect to a \textit{parity-check matrix} and not a \textit{code}. A code can be represented by many parity-check matrices, and for one to be geometrically perfect does not guarantee that all others are. It is shown in \cite{BG,Ka} that a code permits a geometrically perfect parity-check matrix if and only if it does not contain any code equivalent to $\mathcal{H}_7^\perp$, $R_{10}$, or $C(K_5)^\perp$ as a minor. However, it is not clear which parity-check matrix makes that so. Our work is a step toward understanding this choice of representation.

An algebraic characterization of the pseudocodewords is given in \cite{KLVW}. Since $PC(H)$ is closed under addition, its elements must form a cone in $\mathbb{R}^n$. Koetter et al. \cite{KLVW} identify this cone and precisely characterize the elements that are pseudocodewords.

\begin{definition}\label{def:funda}
Let $H\in\mathbb{F}_2^{r\times n}$. The \textit{fundamental cone} of $H$, denoted $K(H)$, is given by
\[K(H)=\left\{\mathbf{v}\in\mathbb{R}^n\:\bigg{|}\:v_i\geq0\textrm{ and }\sum_{l=1,l\neq i}^n{h_{jl}v_l}\geq h_{ji}v_i\textrm{ for all }1\leq i\leq n\textrm{ and }1\leq j\leq r\right\}.\]
\end{definition}

\begin{theorem}\cite[Theorem 4.4]{KLVW}\label{th:pseudo}
Let $H\in\mathbb{F}_2^{r\times n}$. An integer vector $\mathbf{p}$ is a pseudocodeword of $H$ if and only if
\[\mathbf{p}\in K(H)\quad\textrm{and}\quad H\mathbf{p}^T=\mathbf{0}\pmod{2}.\]
\end{theorem}

The inequalities of the fundamental cone demand, roughly speaking, that no single element of a pseudocodeword gets too large. Theorem \ref{th:pseudo} gives a convenient method to enumerate $P(H)$. Instead of going through all covers of the Tanner graph $T(H)$, one may verify whether $\mathbf{p}$ is a pseudocodeword by checking its entries against a set of inequalities and parity conditions determined by the rows of $H$. It is also easy to see that $C(H)\subset K(H)$, and so $K(H)$ contains every integral combination of the codewords.

\section{Pseudocodeword-Free Representation of a Code}

We have seen that the Tanner graph is a convenient graphical representation of a parity-check matrix. Each check node of the Tanner graph literally serves as a ``parity-check'' for a legitimate codeword. On the other hand, determining whether an integer vector is a pseudocodeword requires one to either find a suitable graph cover or appeal to Theorem \ref{th:pseudo} and test the vector algebraically. To assist our study of the pseudocodewords, we first give a new use to the Tanner graph.

Recall that $(c_1,c_2,\ldots,c_n)$ is a codeword of $C(H)$ if and only if the assignment $c_1,c_2,\ldots,c_n$ to the bit nodes of the Tanner graph $T(H)$ makes the binary sum of the neighbors of every check node zero. In other words, a check node is ``satisfied'' if the binary sum of its neighbors is zero, and a codeword is what makes every check node satisfied. Motivated by this observation, we generalize the property of being ``satisfied'' in the following definition.

\begin{definition}\label{def:p-sat}
Let $v$ be a vertex of a graph $G$ with the set of neighbors $\{u_1,u_2,\ldots,u_t\}$. Suppose that integer values $a_1,\ldots,a_t$ are assigned to these vertices. We say that $v$ is \textit{p-satisfied} if all of the the following conditions hold:
\begin{enumerate}[label=\textit{\roman*)}]
\item $a_i\geq0$ for all $i$,
\item $\sum_{l=1}^t{a_l}=0\pmod{2}$, and
\item $\sum_{l=1,l\neq i}^t{a_l}\geq a_i$ for all $1\leq i\leq t$.
\end{enumerate}
\end{definition}

Technically, the property of being p-satisfy localizes the conditions for pseudocodewords given in Theorem \ref{th:pseudo}. It makes the Tanner graph sufficient to verify whether an integer vector is a pseudocodeword, and we state this fact as the following proposition.

\begin{proposition}
Let $T(H)$ be the Tanner graph of $H\in\mathbb{F}_2^{r\times n}$. An integer vector $(p_1,\ldots,p_n)$ is a pseudocodeword if and only if the assignment $p_1,\ldots,p_n$ to the corresponding bit nodes of $T(H)$ makes every check node p-satisfied.
\end{proposition}

\begin{proof}
Combining Definition \ref{def:funda} and Theorem \ref{th:pseudo}, an integer vector $(p_1,\ldots,p_n)$ is a pseudocodeword if and only if, for all $1\leq j\leq r$,
\begin{enumerate}[label=\textit{\roman*)}]
\item $p_i\geq0$ for all $i$,
\item $\sum_{l=1,l\neq i}^n{h_{jl}p_l}\ge h_{ji}p_i$ for all $1\leq i\leq n$, and
\item $\sum_{l=1}^n{h_{jl}p_l}\equiv0\pmod{2}$.
\end{enumerate}
Since $h_{jl}=1$ precisely when the bit node $x_l$ is adjacent to the check node $f_j$, the above three conditions are satisfied for all $1\leq j\leq r$ if and only if every check node is p-satisfied with $p_1,\ldots,p_n$.
\end{proof}

We provide several useful results before stating the main finding of this work. The proof of Lemma \ref{lem:weight} is trivial and is omitted. Here, $w(\cdot)$ denotes the Hamming weight of a binary vector.

\begin{lemma}\label{lem:weight}
Let $H$ be a parity-check matrix. Suppose that the check node $f_{j_1}$ corresponding to $Row_{j_1}(H)$ and the check node $f_{j_2}$ corresponding to $Row_{j_2}(H)$ have no common neighbor. If $\mathbf{r}$ is the binary sum of $Row_{j_1}(H)$ and $Row_{j_2}(H)$, then
\[w(\mathbf{r})=w(Row_{j_1}(H))+w(Row_{j_2}(H)).\]
In particular,
\[w(\mathbf{r})\geq w(Row_{j_1}(H))\quad\textrm{and}\quad w(\mathbf{r})\geq w(Row_{j_2}(H)).\]
\end{lemma}

\begin{proposition}\label{prop:lindep}
Let $H\in\mathbb{F}_2^{r\times n}$. If $T(H)$ is a tree where every check node has degree at least 2, then
\begin{enumerate}[label=\textit{\roman*})]
\item the number of 1's in $H$ is precisely $r+n-1$,
\item if
\[\mathbf{r}=Row_{j_1}(H)+Row_{j_2}(H)+\ldots+Row_{j_s}(H)\pmod{2},\]
then $w(\mathbf{r})\geq w(Row_{j_t}(H))$ for $t=1,2,\ldots,s$, and
\item the rows of $H$ are linearly independent over $\mathbb{F}_2$.
\end{enumerate}
\end{proposition}

\begin{proof}
Since $T(H)$ is a tree with $r+n$ vertices, $T(H)$ has $r+n-1$ edges, and \textit{i)} follows.

Let $\mathbf{r}=Row_{j_1}(H)+Row_{j_2}(H)+\ldots+Row_{j_s}(H)\pmod{2}$. Without loss of generality, we shall prove that $w(\mathbf{r})\geq w(Row_{j_1}(H))$, and \textit{ii)} will readily follow. Consider the tree $T(H)$ with the check node $f_{j_1}$ corresponding to $Row_{j_1}(H)$ as a root node. Clearly, $T(H)\setminus Row_{j_1}(H)$, the subgraph of $T(H)$ with vertex $Row_{j_1}(H)$ removed, is a graph with $\deg(f_{j_1})=w(Row_{j_1}(H))$ connected components. Construct a check node $f_\mathbf{r}$ corresponding to $\mathbf{r}$ (that is, $x_i$ is adjacent to $f_\mathbf{r}$ if and only if the $i^\mathrm{th}$ element of $\mathbf{r}$ is 1). If we can show that $f_\mathbf{r}$ is connected to each component of $T(H)\setminus Row_{j_1}(H)$, then $\deg(f_\mathbf{r})\geq\deg(f_{j_1})$, and we can conclude that $w(\mathbf{r})\geq w(Row_{j_1}(H))$ as required.

Let $U$ be one of the connected components of $T(H)\setminus Row_{j_1}(H)$, and consider the check nodes among $f_{j_2},\ldots,f_{j_s}$ that belong to this component. If there is none, then $f_\mathbf{r}$ is adjacent to the bit node in this component that is adjacent to $f_{j_1}$. Assume without loss of generality now that check nodes $f_{j_2},\ldots,f_{j_k}$ belong to this component. Consider the subgraph of $U$ consisting of check nodes $f_{j_2},\ldots,f_{j_k}$ and their neighbors. This subgraph must be cycle-free and hence has at least two leaves. However, no check node can be a leaf since they have degree at least 2. The check node $f_\mathbf{r}$ must now be adjacent to one of the leaves and therefore connected to this component. This finishes the proof of \textit{ii)}.

Finally, \textit{iii)} follows from \textit{ii)} since we cannot have
\[\mathbf{0}=Row_{j_1}(H)+Row_{j_2}(H)+\ldots+Row_{j_s}(H)\pmod{2}\]
for any rows $j_1,\ldots,j_s$ of $H$.


\end{proof}

Proposition \ref{prop:lindep} will be helpful since we will be dealing with a parity-check matrix whose Tanner graph does not have a cycle. We know that a representation $H$ of a code is geometrically perfect if $T(H)$ is cycle-free. The following theorem now classifies all geometrically perfect representations of such code.

\begin{theorem}\label{th:main}
Let $C$ be a code that has a cycle-free representation. A parity-check matrix $H$ of this code is geometrically perfect if and only if (possibly empty) redundant rows of $H$ can be removed so as to obtain a representation of $C$ that is cycle-free.
\end{theorem}

\begin{proof}
Suppose that $C$ is a concatenation of $e$ codes, and let $H'\in\mathbb{F}_2^{r\times n}$ be any cycle-free representation of $C$. It is not hard to see that $T(H')$ is a forest, i.e., a collection of disconnected $e$ trees. We assume that $T(H')$ has no check node of degree 1, for if $T(H')$ has a check node of degree $1$, then the coordinate corresponding to the neighboring bit node can be punctured or ignored.

Assume that (possibly empty) rows of $H$ can be removed so as to obtain a representation $\widetilde{H}$ of $C$ that is cycle-free. That is, $C(\widetilde{H})=C(H)$, and $T(\widetilde{H})$ is cycle-free. Then, it follows from \cite[Theorem 6.1]{KS} that
\[PC(H)\subseteq PC(\widetilde{H}).\]
Now, since $C(\widetilde{H})=C(H)$, integer combinations of codewords of $\widetilde{H}$ and codewords of $H$ are the same. As $\widetilde{H}$ is geometrically perfect, we have
\[P(H)\subseteq P(\widetilde{H})=\left\{\sum_{c\in C(\widetilde{H})}{a_cc}\mid a_c\in\mathbb{N}\right\}=\left\{\sum_{c\in C(H)}{a_cc}\mid a_c\in\mathbb{N}\right\}\subseteq P(H).\]
Therefore, $H$ is geometrically perfect.

Suppose now that it is not possible to remove redundant rows of $H$ so that a cycle-free representation of $C$ can be obtained. Let $j'_1,j'_2,\ldots,j'_s$ be the rows of $H'$ that are not in $H$. This list is not empty since we cannot remove rows of $H$ to obtain $H'$. Since $H$ and $H'$ represent the same null space (i.e., the code $C$), for each $k=1,2,\ldots,s$, some rows of $H$ must be binary sum of $j'_k$ and a linear combination of some other rows of $H'$; we denote such rows $j_{k,1},\ldots,j_{k,a_k}$. It follows from Lemma \ref{lem:weight} and Proposition \ref{prop:lindep} that
\begin{equation}\label{eq:supp}
w(Row_{j_{k,l}}(H))\geq w(Row_{j'_k}(H'))
\end{equation}
for all $1\leq l\leq a_k$. Suppose for the sake of contradiction that, for each $k=1,2,\ldots,s$, there is $1\leq b_k\leq a_k$ such that $w(Row_{j_{k,b_k}}(H))=w(Row_{j'_k}(H'))$. This means we can remove rows of $H$ so that only rows $j_{1,b_1},j_{2,b_2},\ldots,j_{s,b_s}$ and rows of $H$ that are in $H'$ remain. This representation of $C$ is cycle-free, and hence is a contradiction. Therefore, there exists a row $j'_p$ of $H'$ such that the equality of \eqref{eq:supp} does not hold for any $1\leq l\leq a_p$. This row is essential to our proof, and we call it \textit{pivotal}.

The pivotal row $j'_p$ of $H'$ is represented by a check node $f_{j'_p}$ in the Tanner graph $T(H')$. Let $U'$ be the connected component of $T(H')$ that $j'_p$ belongs to (i.e., $U'$ is one of $e$ trees in the forest $T(H')$). Denote $d:=w(Row_{j'_p}(H'))$ so that the check node $f_{j'_p}$ has degree $d$. It is clear that $U'\setminus f_{j'_p}$, the subgraph of $U'$ with vertex $f_{j'_p}$ removed, is a cycle-free graph with $d$ connected components. We assign values $2d$ to all bit notes in one component and 2 to all other bit nodes of $T(H')$. We will prove that this assignment does not p-satisfy $f_{j'_p}$ but p-satisfies every check node of $T(H)$. Hence, we can appeal to Proposition \ref{prop:lindep} and conclude that $H$ is not geometrically perfect as $PC(H)$ contains an extraneous pseudocodeword.

\begin{figure}
\begin{center}
\includegraphics[width=15cm]{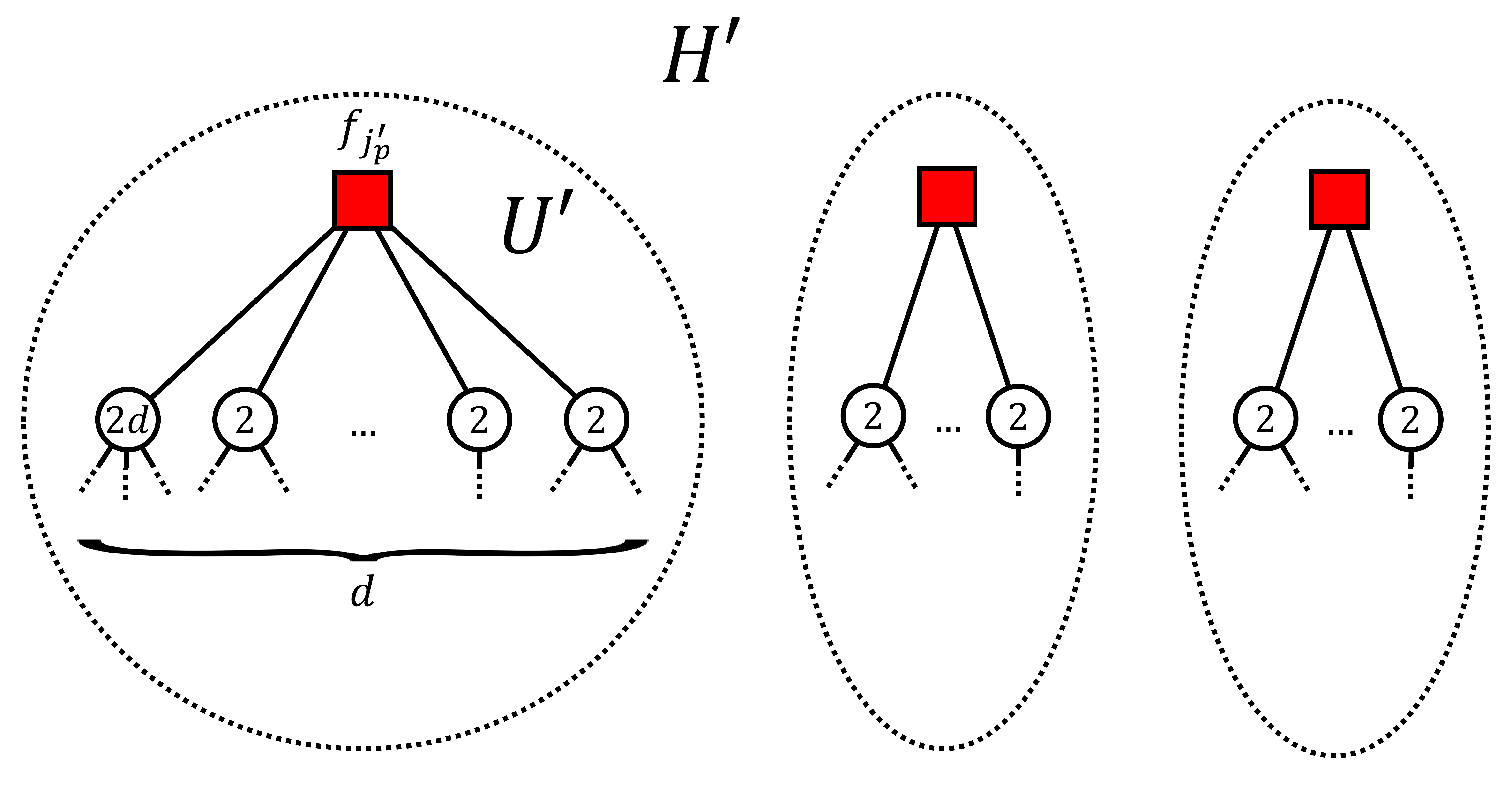} \\[20pt]
\includegraphics[width=15cm]{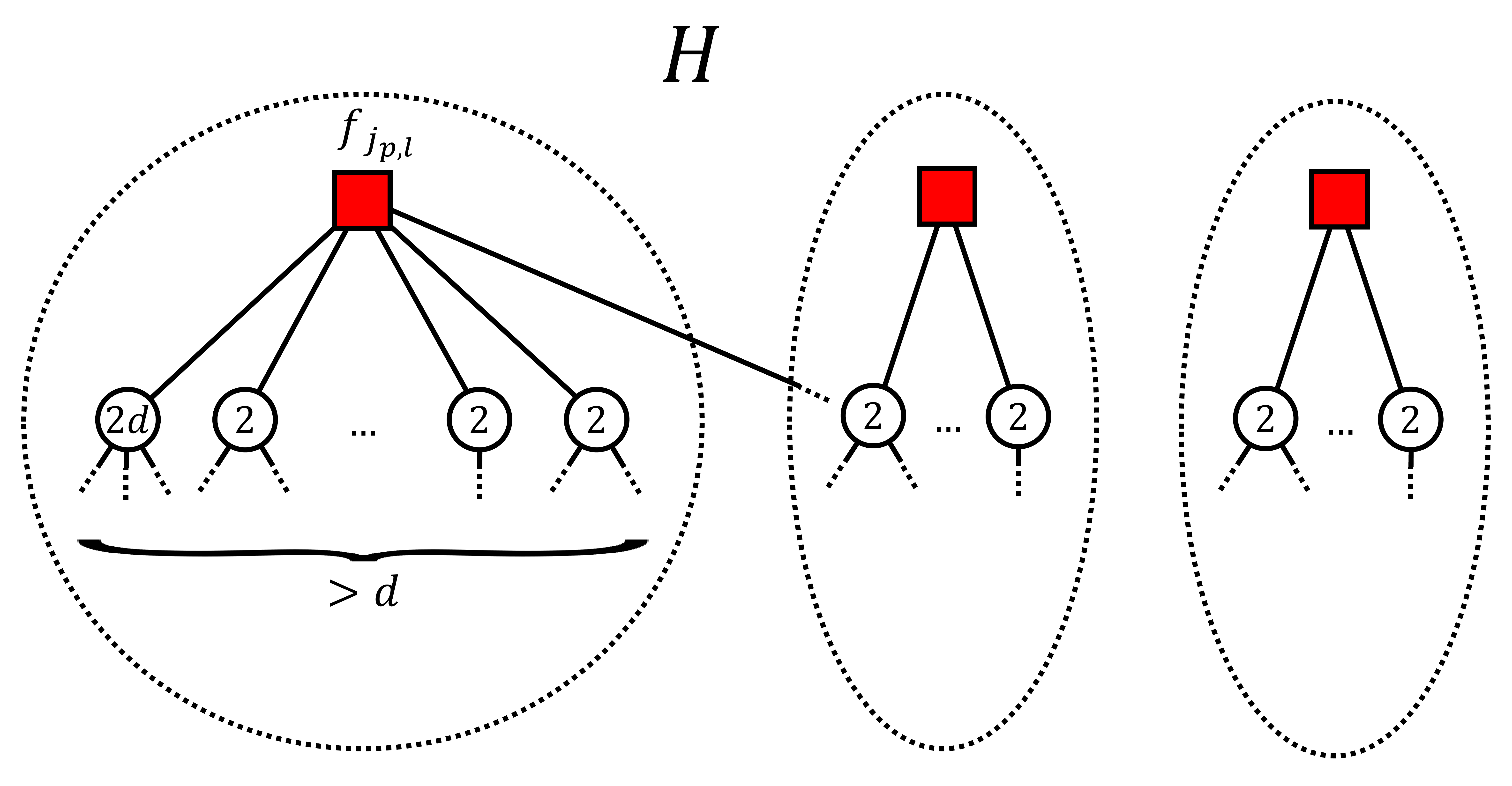}
\caption{\label{fig:proof} Check node $f_{j'_p}$ is not p-satisfied (top) while check node $f_{j_{p,l}}$ is (bottom).}
\end{center}
\end{figure}

Our arguments will center around condition \textit{iii)} of Definition \ref{def:p-sat}. Roughly speaking, this condition requires that the value at one neighboring bit node cannot be greater than the values at all other neighboring bit nodes combined. The check node $f_{j'_p}$ of $H'$ is adjacent to 1 bit node whose value is $2d$ and $d-1$ bit nodes whose value is $2$ (see Figure \ref{fig:proof}), and so this check node is not p-satisfied as $\underbrace{2+2+\ldots+2}_{d-1\textrm{ terms}}<2d$.

We are left to show that the same assignment makes a pseudocodeword for $H$. Conditions \textit{i)} and \textit{ii)} from Definition \ref{def:p-sat} are trivially satisfied for every check node of $T(H)$, and now only condition \textit{iii)} is to be verified. Rows of $H$ can be identified as either
\begin{itemize}
\item $j_{p,1},\ldots,j_{p,a_p}$, or
\item linear combination of rows $1,2,\ldots,j'_p-1,j'_p+1,\ldots,r$ from $H'$.
\end{itemize}
Since
\[w(Row_{j_{p,l}}(H))>w(Row_{j'_p}(H'))=d\]
for all $1\leq l\leq a_p$, the required condition is satisfied for check nodes $j_{p,1},\ldots,j_{p,a_p}$ of $H$ (see Figure \ref{fig:proof}).

Consider now a check node that is a linear combination of rows $1,2,\ldots,j'_p-1,j'_p+1,\ldots,r$ of $H'$. If this check node is only adjacent to bit nodes whose value is 2, then clearly condition \textit{iii)} is satisfied. Finally, if this check node is adjacent to a bit node whose value is $2d$, it must be adjacent to at least 2 such bit nodes since each check node has degree at least 2. Thus, conditions \textit{iii)} is satisfied, and this finishes the proof of the theorem.
\end{proof}

\section{Examples}

We present two examples in this section. The first one illustrates the key step used in the proof of Theorem \ref{th:main}. The second one showcases the code considered in \cite{W}. As we will see, it is possible that a parity-check matrix is geometrically perfect despite the presence of several small cycles.

\begin{figure}
\begin{center}
\includegraphics[width=17cm]{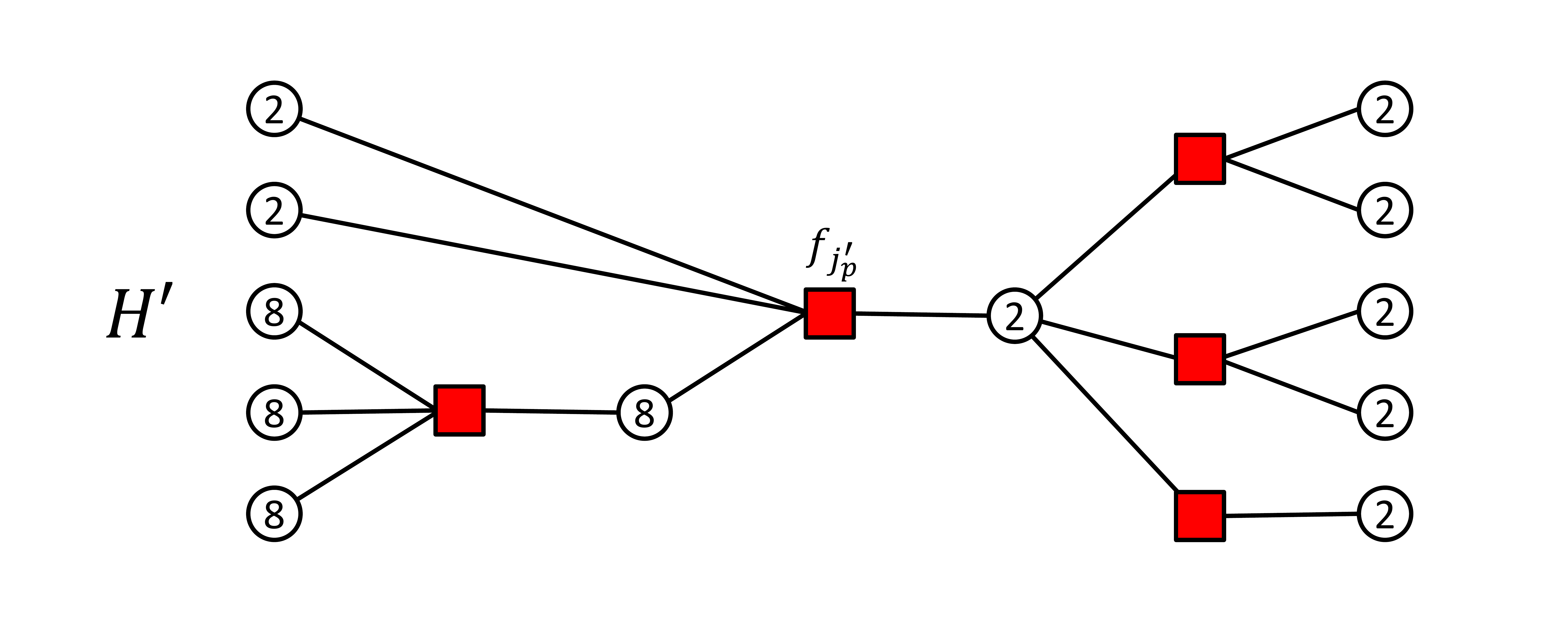} \\
\includegraphics[width=17cm]{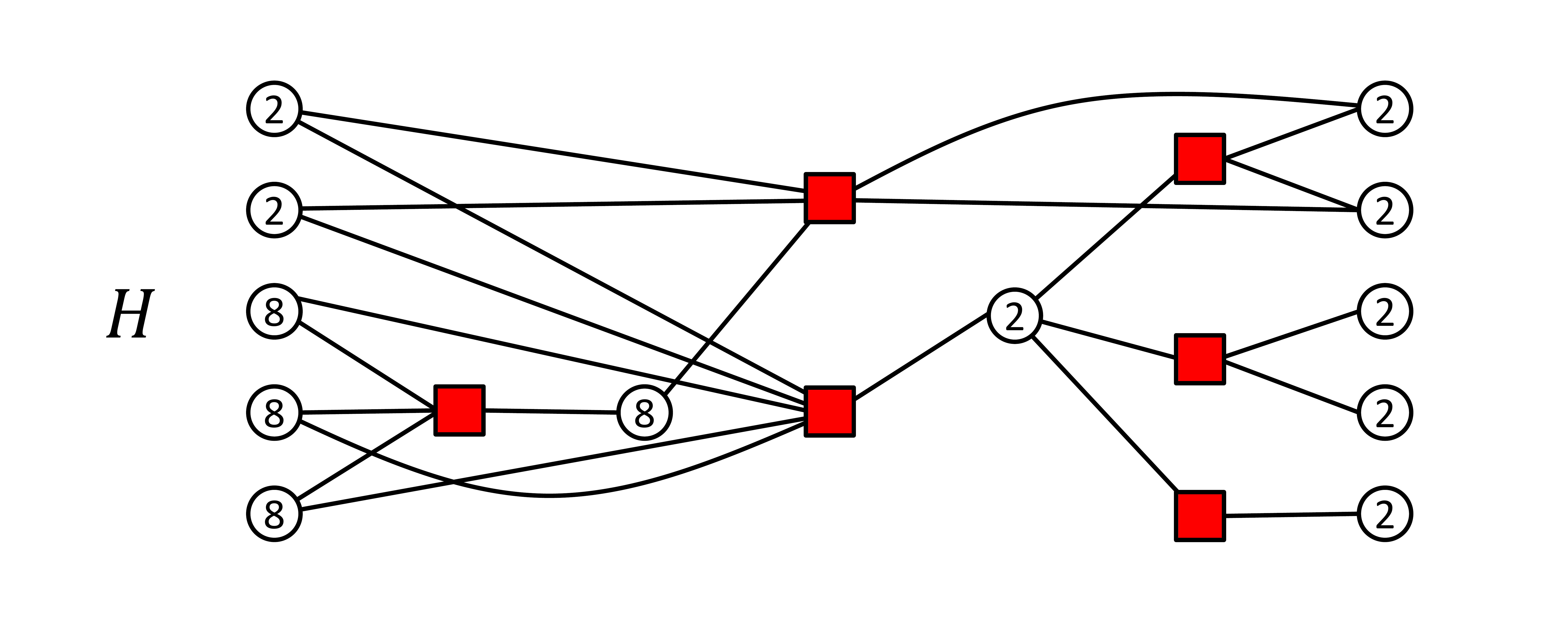}
\caption{\label{fig:example} On the top presents the Tanner graph $T(H')$ and the pivotal check node $f_{j'_p}$ from Example \ref{ex:proof}. The check node in consideration is not p-satisfied. At the bottom is the Tanner graph of $H$. One readily sees that every check node is p-satisfied, and so $(2,2,8,8,8,8,2,2,2,2,2,2)$ is a pseudocodeword for $H$.}
\end{center}
\end{figure}

\begin{example}\label{ex:proof}
Suppose that
\[H'=\left(
\begin{array}{ccccccccccccc}
0 & 0 & 1 & 1 & 1 & 1 & 0 & 0 & 0 & 0 & 0 & 0 \\
1 & 1 & 0 & 0 & 0 & 1 & 1 & 0 & 0 & 0 & 0 & 0 \\
0 & 0 & 0 & 0 & 0 & 0 & 1 & 1 & 1 & 0 & 0 & 0 \\
0 & 0 & 0 & 0 & 0 & 0 & 1 & 0 & 0 & 1 & 1 & 0 \\
0 & 0 & 0 & 0 & 0 & 0 & 1 & 0 & 0 & 0 & 0 & 1
\end{array}\right)\]
and
\[H=\left(
\begin{array}{ccccccccccccc}
0 & 0 & 1 & 1 & 1 & 1 & 0 & 0 & 0 & 0 & 0 & 0 \\
1 & 1 & 1 & 1 & 1 & 0 & 1 & 0 & 0 & 0 & 0 & 0 \\
1 & 1 & 0 & 0 & 0 & 1 & 0 & 1 & 1 & 0 & 0 & 0 \\
0 & 0 & 0 & 0 & 0 & 0 & 1 & 1 & 1 & 0 & 0 & 0 \\
0 & 0 & 0 & 0 & 0 & 0 & 1 & 0 & 0 & 1 & 1 & 0 \\
0 & 0 & 0 & 0 & 0 & 0 & 1 & 0 & 0 & 0 & 0 & 1
\end{array}\right).\]
Let $C=C(H')=C(H)$. It is not possible to remove rows of $H$ to obtain $H'$ or any cycle-free representation of $C$. Here, the second row of $H'$ is pivotal and is labeled $f_{j'_p}$ in Figure \ref{fig:example}. We have $d=w(Row_2(H'))=4$ and $T(H')\setminus f_{j'_p}$ has 4 connected components. Value 8 is assigned to all bit notes in one component, and value 2 is assigned to all other bit nodes. The check node $f_{j'_p}$ of $T(H')$ is not p-satisfied with this assignment. However, every check node of $T(H)$ is p-satisfied. This makes
\[PC(H)\subset PC(H'),\]
and so $H'$ is not geometrically perfect.
\end{example}

\begin{figure}
\begin{center}
\includegraphics[width=7cm]{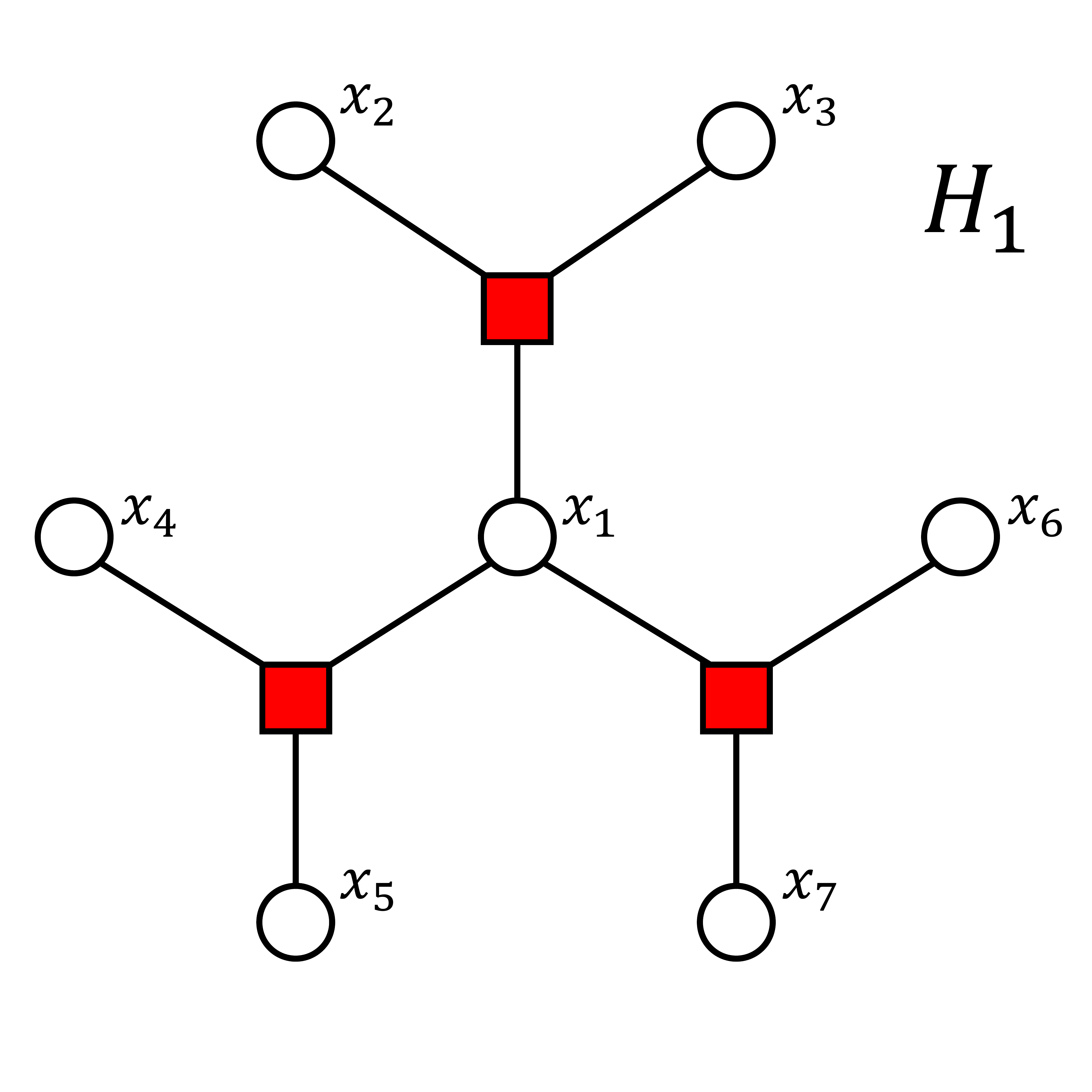} \quad
\includegraphics[width=7cm]{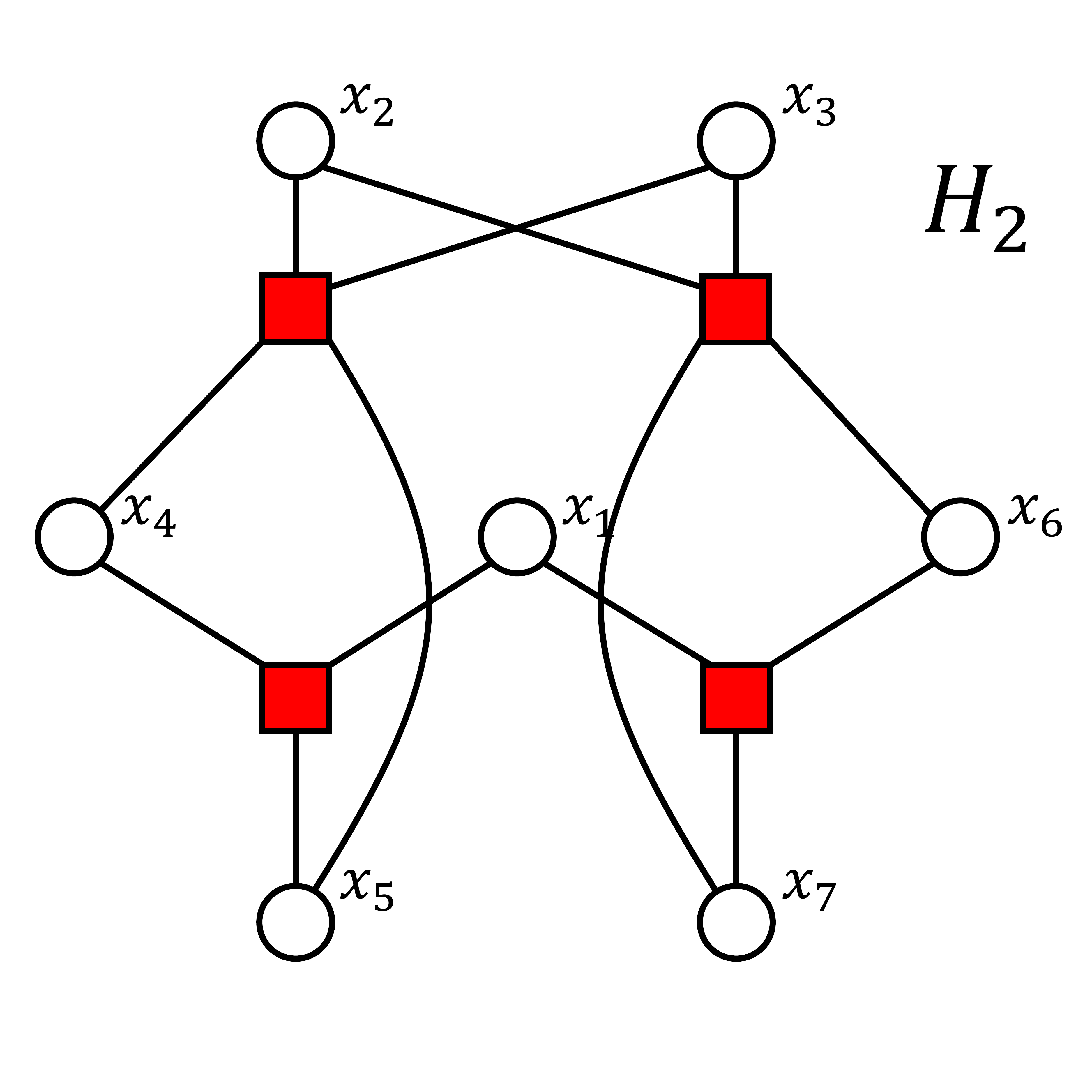} \\
\includegraphics[width=7cm]{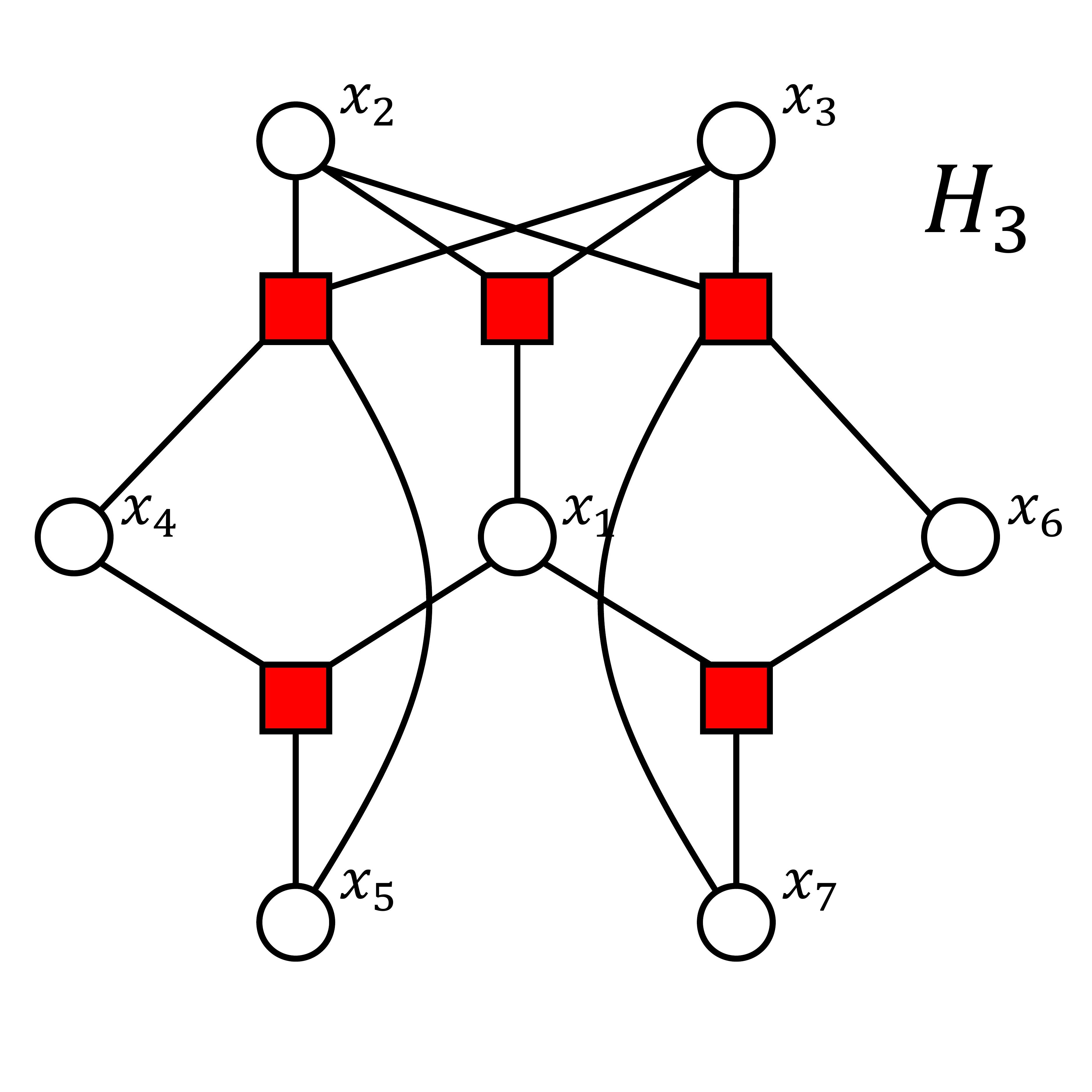}
\caption{\label{fig:wib} Three representations of the the code from Example \ref{ex:wib}}
\end{center}
\end{figure}

\begin{example}\label{ex:wib}
This example considers three representations of the code of length 7 and dimension 4 used to demonstrate the min-sum algorithm in \cite[Section 3.1]{W}. The first representation $H_1$ yields a tree, the second representation $H_2$ is a cycle code (i.e., every bit node has degree 2), and the third representation combines check conditions of the two. See Figure \ref{fig:wib}.

It is not hard to see that
\[C=\langle(1,1,0,1,0,1,0),(0,1,1,0,0,0,0),(0,0,0,1,1,0,0),(0,0,0,0,0,1,1)\rangle.\]
Now, $H_1$ is geometrically perfect, and so
\[PC(H_1)=\left\{\sum_{\mathbf{c}\in C}{a_\mathbf{c}\mathbf{c}\:\bigg{|}\:a_\mathbf{c}\in\mathbb{N}}\right\}.\]
On the other hand, there is no redundant row of $H_2$ whose removal yields a cycle-free representation of $C$. This representation is not geometrically perfect as
\[PC(H_2)=\left\{\sum_{\mathbf{d}\in D}{a_\mathbf{d}\mathbf{d}\:\bigg{|}\:a_\mathbf{d}\in\mathbb{N}}\right\}\]
where
\[D=C\cup\{(2,0,0,1,1,1,1),(0,2,0,1,1,1,1),(0,0,2,1,1,1,1),(2,2,2,1,1,1,1)\}.\]
The third representation combines the rows of $H_1$ and $H_2$, and is quite interesting since it contains both a subgraph that is a tree and several 4-cycles. In this case, the tree dominates as the representation $H_3$ is geometrically perfect. In other words, there are no pseudocodewords besides the integral combinations of the codewords. One plausible explanation here in terms of graph cover is that, although small cycles bolster the pseudocodewords, the cycle-free subgraph $T(H_1)$ forbids one from coming up.
\end{example}

\section*{Acknowledgments}
The author wishes to thank Gretchen L. Matthews for her support while the author is at Clemson University and Patanee Udomkavanich for her advice. This work is supported by the Thailand Research Fund under Grant TRG5880116 and the Centre of Excellence in Mathematics, the Commission on Higher Education, Thailand.


\begin{thebibliography}{99}
\bibitem{Axvig}
N. Axvig, D. Dreher, K. Morrison, E. Psota, L. C. Perez, and J. L. Walker, Analysis of connections between pseudocodewords, IEEE Trans. Inform. Theory \textbf{55} (2009), no. 9, 4099--4107.
\bibitem{BG}
F. Barahona and M. Gr\"otschel, On the cycle polytope of a binary matroid, J. Comb. Theory, Ser. B \textbf{40} (1986), 40--62.
\bibitem{ETV}
T. Etzion, A. Trachtenberg, A. Vardy, Which codes have cycle-free Tanner graphs?, IEEE Trans. Inform. Theory \textbf{45} (1999), no. 6, 2173--2181.
\bibitem{FWK}
J. Feldman, M. J. Wainwright, and D. R. Karger, Using linear programming to decode binary linear codes, IEEE Trans. Inform. Theory \textbf{51} (2005), no. 3, 954--972.
\bibitem{Ka}
N. Kashyap, A decomposition theory for binary linear codes, IEEE Trans. Inform. Theory \textbf{54} (2008), no. 7, 3035--3058.
\bibitem{KS}
C. Kelley and D. Sridhara, Pseudocodewords of Tanner graphs, IEEE Trans. Inform. Theory \textbf{53} (2007), no. 11, 4013--4038.
\bibitem{KLVW}
R. Koetter, W.-C. W. Li, P. O. Vontobel, and J. Walker,  Characterizations of pseudo-codewords of (low-density) parity-check codes, Adv. Math. \textbf{213} (2007), no. 1, 205--229.
\bibitem{KM}
W. Kositwattanarerk and G. L. Matthews, Lifting the fundamental cone and enumerating the pseudocodewords of a parity-check code, IEEE Trans. Inform. Theory \textbf{57} (2011), no. 2, 898--909.
\bibitem{KLF}
Y. Kou, S. Lin, and M. P. C. Fossorier, Low-density parity-check codes based on finite geometries: a rediscovery and new results, IEEE Trans. Inform. Theory \textbf{47} (2001), no. 7, 2711--2736.
\bibitem{KFL}
F. R. Kschischang, B. J. Frey, and H.-A. Loeliger, Factor graphs and the sum-product algorithm, IEEE Trans. Inform. Theory \textbf{47} (2001), no. 2, 498--519.
\bibitem{L}
G. Lechner, The effect of cycles on binary message-passing decoding of LDPC codes, Proc. Comm. Theory Workshop, Australia, IEEE (2010).
\bibitem{MN}
D. J. C. MacKay and R. M. Neal, Near Shannon limit performance of low density parity check codes, Electronics Letters \textbf{32} (1996), 1645--1646.
\bibitem{RSU}
T. Richardson and A. Shokrollahi and R. Urbanke, Design of capacity-approaching irregular low-density parity-check codes, IEEE Trans. Inform. Theory \textbf{47} (2001), no. 2, 619--637.
\bibitem{SV}
R. Smarandache and P. O. Vontobel, Pseudo-codeword analysis of Tanner graphs from projective and Euclidean planes, IEEE Trans. Inform. Theory \textbf{53} (2007), no. 7, 2376--2393.
\bibitem{Tanner}
R. M. Tanner, A recursive approach to low-complexity codes, IEEE Trans. Inform. Theory \textbf{27} (1981), 533--547.
\bibitem{Tian}
T. Tian, C. R. Jones, J. D. Villasenor, and R. D. Wesel, Selective avoidance of cycles in irregular LDPC code construction, IEEE Trans. Comm. \textbf{52} (2004), 1242--1247.
\bibitem{W}
N. Wiberg, Codes and decoding on general graphs, Ph.D. thesis, Link\"oping University, Link\"oping, Sweden, 1996.
\bibitem{XF}
S.-T. Xia and F.-W. Fu, Minimum pseudoweight and minimum pseudocodewords of LDPC codes, IEEE Trans. Inform. Theory \textbf{54} (2008), 480--185.
\bibitem{Xu}
J. Xu, L. Chen, I. Djurdjevic, S. Lin, K. Abdel-Ghaffar, Construction of regular and irregular LDPC codes: geometry decomposition and masking, IEEE Trans. Inform. Theory \textbf{53} (2007), 121--134.
\bibitem{ZSF}
J. Zumbragel, V. Skachek, and M. F. Flanagan, On the pseudocodeword redundancy of binary linear codes, IEEE Trans. Inform. Theory \textbf{58} (2012), 4848--4861.
\end{thebibliography}
\end{document}